\documentclass{tJBD2e}

\usepackage{epstopdf}% To incorporate .eps illustrations using PDFLaTeX, etc.
\usepackage{subfigure}% Support for small, `sub' figures and tables

\theoremstyle{plain}% Theorem-like structures
\newtheorem{theorem}{Theorem}[section]

\theoremstyle{definition}

\theoremstyle{remark}

 \usepackage{xcolor}
\usepackage{amsfonts,amsmath,amsthm}
\usepackage{graphicx}
\usepackage{epstopdf}
\usepackage{subfigure}
\usepackage{amsfonts,amssymb}
\usepackage{amsmath,amscd}
\usepackage{amsmath}
\usepackage{listings, fancyhdr,subfigure}
\newcommand\dif{\mathrm{d}}
\usepackage{mathrsfs}

\usepackage[all]{xy}
\def\bc{\begin{center}}       \def\ec{\end{center}}
\def\ba{\begin{array}}        \def\ea{\end{array}}
\def\be{\begin{equation}}     \def\ee{\end{equation}}
\def\be*{\begin{equation*}}     \def\ee*{\end{equation*}}
\def\bea{\begin{eqnarray}}    \def\eea{\end{eqnarray}}
\def\bea*{\begin{eqnarray*}}  \def\eea*{\end{eqnarray*}}

\newif\ifstylew  \stylewtrue
\newcommand{\style}[1]{\ifstylew \textcolor{green}{  $\spadesuit$ }\ {\sf \bf \it  #1}\ \textcolor{green}{ $\spadesuit$ } \fi}
\stylewfalse   % comment this line out to turn on style notes

\newif\ifnotesw \noteswtrue
\newcommand{\notes}[1]{\ifnotesw \textcolor{red}{  $\clubsuit$\ {\sf \bf \it  #1}\ $\clubsuit$ }\fi}
\begin{document}

%\jvol{00} \jnum{00} \jyear{2015} \jmonth{March}

%\articletype{GUIDE}

\title{Two-sex mosquito model for the persistence of \emph{Wolbachia}}

\author{
\name{Ling Xue \textsuperscript{a}$^{\ast}$\thanks{$^\ast$Corresponding author. Email: lxue2@tulane.edu}
and Carrie A. Manore\textsuperscript{a} and Panpim Thongsripong \textsuperscript{b} and James M. Hyman\textsuperscript{a}}
\affil{\textsuperscript{a}  Department of Mathematics, Center for Computational Science, \\Tulane University, New Orleans, LA 70118;\\
\textsuperscript{b}School of Public Health and Tropical Medicine,\\ Tulane University, New Orleans, LA 70112  }
%\received{v5.0 released March 2015}
}

\maketitle

\begin{abstract}
We  develop and analyze an ordinary differential equation   model to investigate
the transmission dynamics of releasing \emph{Wolbachia}-infected mosquitoes to establish an endemic
infection in a population of wild uninfected mosquitoes.
\emph{Wolbachia} is a genus of  endosymbiotic bacteria that can infect mosquitoes and reduce their ability to transmit dengue virus.
Although the bacterium is transmitted vertically from infected mothers to their offspring,
it can be difficult to establish an endemic  infection in a wild mosquito population.
Our  transmission model for the adult and aquatic-stage mosquitoes takes into account \emph{Wolbachia}-induced fitness change and cytoplasmic incompatibility.
We show that, for a wide range of realistic parameter values, the basic reproduction number, $R_0$, is less than one. Hence, the epidemic will die out if only a few \emph{Wolbachia}-infected mosquitoes  are introduced into the wild population.
Even though the basic reproduction number is less than one, an endemic \emph{Wolbachia} infection can be established if a sufficient number of infected mosquitoes are released.
This threshold effect is created by a backward bifurcation with three coexisting equilibria: a stable zero-infection equilibrium, an intermediate-infection unstable endemic equilibrium, and a high-infection stable endemic equilibrium.
We analyze the impact of reducing the wild mosquito population before introducing the infected mosquitoes
and observed that the most effective approach to establish the infection in the wild is based on reducing mosquitoes in both the adult and aquatic stages.
\end{abstract}


\section{Introduction} \label{Introduction}

 We use a disease transmission  model to investigate  the
 conditions for releasing \emph{Wolbachia}-infected mosquitoes that will establish
 an endemic infection in a population of wild uninfected mosquitoes.  The
 \emph{Wolbachia} infected mosquitoes are less able to  transmit dengue virus, and
the goal of the modeling effort is to better understand how this bacterium can be used
to control vector-borne diseases.  We observed that the basic
reproductive number for the \emph{Wolbachia} model is less than one for typical model
parameters, and so small \emph{Wolbachia} infestations will die out. However,
the model predicts that there is a critical threshold, and if a sufficient number of
infected mosquitoes are released, then an endemic \emph{Wolbachia}-infected
population of mosquitoes can be established.  We also investigated how
 this critical threshold can be reduced by first decreasing the population of wild mosquitoes.

Dengue  is  the world's most significant and widespread arthropod-borne viral disease \cite{Gibbons2002}. Each year, $400$ million people are infected with dengue virus in more than $100$ countries \cite{Kyle2008}, while the other one-third of the world's population is at risk. To date, there are no vaccines or specific therapy  available. Traditional control strategies focusing on  reducing population of \emph{Aedes} mosquito vectors have failed to slow the current dengue pandemic, especially in tropical communities \cite{Walker2011}.  The main vector,  \emph{Aedes aegypti},  rebounds in many areas and the secondary vector,   \emph{Aedes albopictus}  keeps   expanding its geographic distribution, leading to $30$ fold increase in cases over the past $50$ years \cite{Lam2011}, which necessitates effective novel alternatives to break dengue transmission cycles \cite{Walker2011}  targeting  \emph{Aedes aegypti} and \emph{Aedes albopictus}.

Increasing attention has been paid to controlling the spread of dengue by targeting mosquito longevity by introducing genetically modified mosquitoes or introducing endosymbiotic \emph{Wolbachia}  bacteria to shorten the mosquito lifespan \cite{Blagrove2012, McMeniman2009, Walker2011}.
%Increasing attention has been paid to  \emph{Wolbachia}  due to its potential as a control strategy on dengue  and many other vector-borne diseases.  \emph{Wolbachia  pipientis}   bacteria can be passed from female mosquitoes by vertical transmission.  They were found to be causal agent of cytoplasmic incompatibility.\\
That is,  \emph{Wolbachia}-infected mosquitoes are released to create a sustained infection in the wild (uninfected) population. If the infection is sustained, then the wild infected mosquitoes will be less effective in transmitting dengue fever.
We create and analyze a mathematical model to help understand the underlying dynamics of \emph{Wolbachia}-infected mosquitoes that are needed to create a sustained endemic \emph{Wolbachia} infection.  Once the \emph{Wolbachia} disease transmission model is well understood, one of our future goals will be to couple this model with a mosquito-human model for the spread of dengue.

\subsection{Wolbachia Bacteria}
The wMel strain of \emph{Wolbachia pipientis}   bacteria is a  maternally inherited endosymbiont infecting  more than $60
\%$ of all insect species.
This strain  has the ability to alternate host reproduction through parthenogenesis, which results in the development of unfertilized eggs, male killing, feminization, and cytoplasmic incompatibility (CI) \cite{Bourtzis1998, Neil1997} that prevents the eggs from forming viable offspring. The latter includes strategies for both the suppression and replacement of medically important mosquito populations. Cytoplasmic incompatibility  is an incompatibility
between the sperms and eggs induced by  \emph{Wolbachia} infection  \cite{Hoffmann1997} and has received considerable attention as a method to control vector-borne diseases \cite{Suh2013}.     \emph{Wolbachia}  even induces resistance to dengue virus in  \emph{Aedes aegypti}  \cite{Bian2010} and  limits transmission of dengue virus in  \emph{Aedes albopictus}  \cite{Mousson2012}.

 Uninfected females only mate successfully with uninfected males, while infected females can mate successfully with both uninfected and infected males \cite{Dobson2002}.    If a male fertilizes a female harboring the same type of infection, the offspring still can survive \cite{McMeniman2009}.
  When \emph{Wolbachia}-infected males mate with uninfected females, or females infected with a different  \emph{Wolbachia}  strain,
  then the CI often results in killing the embryos.  \cite{Hoffmann1997}.   Therefore, infected females have a reproduction advantage over uninfected females due to protection from CI \cite{Weeks2007}.

To  successfully transmit dengue virus, a vector must
imbibe virus particles during blood-feeding and survive
to the point that the pathogen can be biologically transmitted  to
the next vertebrate host  \cite{Rasgon2003}. This time period, called
extrinsic incubation period (EIP), varies with  ambient temperature, many other climatic factors, and characteristics
of the vector-parasite system \cite{Watts1987}.  Vectors that  survive long enough to transmit the
pathogen are called effective vectors \cite{Rasgon2003}. Typically dengue  virus has an incubation period as long as two weeks  to  transmit through \emph{Aedes aegypti}  populations \cite{Gibbons2002}. A life-shortening strain of  \emph{Wolbachia}  may halve the life span of  \emph{Aedes aegypti}  \cite{Moreira2009}.   \emph{Wolbachia}  infection may reduce  the rate of disease transmission  due to the reduction on the lifespan of infected mosquitoes or the interference with mosquito susceptibility to dengue virus.

\subsection{Existing Mosquito-Wolbachia Models}

Ordinary differential equation (ODE) models have been developed to explore key factors that determine the success of applying  \emph{Wolbachia}  to dengue control.
A single-sex model for  \emph{Wolbachia}  infection with both age-structured and unstructured models were presented to study the stability and equilibrium based on the assumption that  \emph{Wolbachia}  infection leads to increased mortality or reduced birth rate \cite{ Farkas2010}. A model assuming a fixed ratio of females and males addressed how pathogen protection affects  \emph{Wolbachia}  invasion \cite{Souto-Maior2014}. Age-structured and unstructured models combining males and females were found to be different in terms of existence and stability of equilibrium solutions \cite{ Farkas2010}.  A stochastic model for female mosquitoes was developed to  investigate the impact of introduction  frequency on establishment of  \emph{Wolbachia}
\cite{Jansen2008}.

Discrete-time models explored the impact of the type of immigration and the temporal dynamics of the host population  on the spread of \emph{Wolbachia}, assuming equal sex ratio between males and females \cite{Hancock2011}.
Discrete generation models for female mosquitoes were built to understand  unstable equilibrium  produced by reduced lifespan or lengthened development \cite{Turelli2010}.   Reaction-diffusion and integro-difference equation
model has been used to analyze the impact of  insect dispersal and infection spread on invasion of   \emph{Wolbachia}  \cite{Schofield2002}.

An ordinary differential equation (ODE) model
 was developed to  evaluate the desirable properties
of the   \emph{Wolbachia}  strain to be introduced to female mosquitoes,  assuming that \emph{Wolbachia}-infected mosquitoes have reduced lifespan and reduced capability to transmit dengue, and equal fraction of male and female mosquitoes \cite{Hughes2013}. A continuous time non-spatial model and  a reaction-diffusion model incorporating lifespan shortening and CI  were developed to study factors that determine the spatial spread of
 \emph{Wolbachia}  through a population of female  \emph{Aedes aegypti}  mosquitoes  assuming constant population size and perfect maternal transmission of \emph{Wolbachia} \cite{Schraiber2012}.  A two-sex deterministic model with deterministic
immature life stages and  stochastic female adult life stage  was developed to  understand  \emph{Wolbachia}  invasion into uninfected host population  \cite{Crain2011}.

A single strain model, two strain model, and spatial model were developed to study whether multi-stain of  \emph{Wolbachia}  can coexist in a spatial context \cite{Keeling2003}. A two-sex ODE model taking into account different death rates, but the same egg laying rates  of \emph{Wolbachia}-infected and \emph{Wolbachia}  uninfected mosquitoes \cite{Koiller2014} showed the basic reproduction number  is always less than one, and the complete infection equilibrium is locally asymptotically stable (LAS) due to positive determinant of the Jacobian matrix for the system.  Simulations showed that dengue epidemics will not occur when \emph{Wolbachia}   infection is sufficiently prevalent \cite{Koiller2014}.  A  two-sex mosquito model assuming  complete vertical transmission and equal death rates for male and female mosquitoes was  developed and four steady states were found \cite{Ndii2012}.

Most of these models  consider either a single-sex model for adult mosquitoes, or assume a fixed ratio between the number of male and female mosquitoes.  Also, most of the models assume  homogeneous death rates and egg laying rates for \emph{Wolbachia}-free and \emph{Wolbachia}-infected mosquitoes. Our model addresses both of these issues.
Our emphasis is to understand how \emph{Wolbachia} infection can be established in a wild population
of mosquitoes.

\iffalse
\notes{What should we say about:  In laboratory-based studies, \emph{Wolbachia}-infected female mosquitoes were found to live longer, produce more eggs, and have higher egg hatching rate than uninfected females by examining superinfection on  \emph{Aedes albopictus}  \cite{Dobson2002, Dobson2004, Mains2013}. However, in one study, the  survival rate of infected larvae is higher when $50$ larvae live together, while lower when $200$ larvae live together in laboratory-controlled condition \cite{Gavotte2009}.
Whether death rates and egg laying rates of \emph{Wolbachia}-infected  mosquitoes  are higher than \emph{Wolbachia}-free  mosquitoes  may depend on   climatic factors and density etc. }  \notes{	This sentence means: sometimes the death rate and egg laying rate of \emph{Wolbachia}-infected female mosquitoes are larger, sometimes smaller, depending on climatic factors and density of mosquitoes.Ling}
\fi
\subsection{Results}
We proposed a compartmental two-sex model to investigate the underlying mechanisms that may contribute to invasion and sustainable establishment of  \emph{Wolbachia}  in mosquito populations. We assigned female and male mosquitoes to different classes to understand corresponding roles that they are playing in the spread of  \emph{Wolbachia} in mosquito populations.

We showed that, for a wide range of realistic parameter values, the basic reproduction number, $R_0$, for this model is less than one. Hence, the epidemic will die out if only a few \emph{Wolbachia}-infected mosquitoes  are introduced into the wild population.
Even though the basic reproduction number is less than one, an endemic \emph{Wolbachia} infection can be established if a sufficient number of infected mosquitoes are released.
This threshold effect can be explained as  a backward bifurcation with three coexisting equilibria: a stable zero-infection equilibrium, an intermediate-infection unstable endemic equilibrium, and a high-infection stable endemic equilibrium.

If the number of infected individuals is below the unstable endemic equilibrium, then the infection decays to the zero-infection equilibrium.  Conversely, if the number of infected mosquitoes is greater than unstable endemic equilibrium, then the solution tends to the stable high-infection equilibrium.
 We identified the relationships between dimensionless combinations of model parameters and the initial conditions for \emph{Wolbachia} to be attracted to the high-infection state.
As expected, the  number of infected female mosquitoes needed to be released to establish the infection in a wild \emph{Wolbachia}-free population decreases as $R_0$ increases to one.
We analyzed the impact of reducing the wild mosquito population before introducing the infected mosquitoes.
We found that the most effective approach of reducing the number of infected mosquitoes needed to establish a wild \emph{Wolbachia}-infected population requires reducing wild mosquito populations in both the adult and aquatic stages before the release.  This could be accomplished by recursive spraying, or a combination of spraying and larvae control.

Our main findings are:
\begin{enumerate}
\item  Three equilibria, disease free equilibrium (DFE), endemic equilibrium (EE), and complete infection  equilibrium (CIE) coexist when $R_0<1$. Disease free equilibrium is a steady state when all individuals are \emph{Wolbachia}-free. Endemic equilibrium is a steady state when some individuals are \emph{Wolbachia}-free, the rest are infected with \emph{Wolbachia}. Complete infection equilibrium is a steady state when all individuals  are infected with \emph{Wolbachia}.
\item The backward bifurcation analysis of our \emph{Wolbachia} transmission model predicts that if $R_0<1$, then
there is a critical threshold for the number of infected mosquitoes released in the wild before the infection can be established. If we release too few \emph{Wolbachia}  infected mosquitoes, then \it Wolbachia \rm infection will die out.
\item Killing both aquatic state (eggs and larvae) and adult mosquitoes before releasing the infected mosquitoes greatly increases the chance that the infection will be established.  Our model quantifies  what fraction of wild mosquitoes must be killed before releasing \emph{Wolbachia}  infected  mosquitoes.

\end{enumerate}

After introducing the mathematical model, we summarize the key results  from the  analysis and numerical simulations.  We conclude with a  discussion of the relevance,  importance, and future directions for this work.

%********************* Methods **********************
\section{Description of Model Framework}  \label{Methods}

 We developed an ODE model incorporating adult females (F),  adult males (M), and
an aggregated aquatic (A) stage that includes the egg, larvae, and pupae stages. The population dynamics of mosquitoes without taking into account \emph{Wolbachia} is in the Appendix, Equation \ref{E:mosquitopopulationmodel}.
The vertical transmission of  \emph{Wolbachia} from infected females to their offspring is a key factor
in establishing an endemic infected population.

Mosquitoes are grouped into six compartments: susceptible   aquatic stage,  $A_u$, infected  aquatic stage,   $A_w$, susceptible female mosquitoes, $F_u$, infected female mosquitoes, $F_w$, susceptible male mosquitoes, $M_{u}$, and   infected male mosquitoes, $M_{w}$.
The eclosion rates of susceptible female and male mosquitoes hatching from eggs are $ \psi\theta A_u$ and $\psi (1-\theta)A_u$, respectively. Similarly, the birth rates of infected female and male mosquitoes are $\psi \theta A_w$ and $\psi (1-\theta)A_w$. Death rates of uninfected male mosquitoes  and infected male mosquitoes are $\mu_{a}A_u$ and $\mu_{a}A_w$. Death rates of uninfected female mosquitoes  and infected female mosquitoes are $\mu_{fu}F_u$ and $\mu_{fw}F_w$. Death rates of uninfected male mosquitoes  and infected male mosquitoes are $\mu_{mu}M_u$ and $\mu_{mw}M_w$. Development rates of uninfected aquatic stage and infected aquatic stage of mosquitoes are $\psi A_u$ and $\psi A_w$.   The model parameters are described in Table \ref{T:statevarandparameters}.

The model (Fig. \ref{Fig:diagram})
 describing population dynamics of aquatic stage, adult male, and adult female mosquitoes is given by:
%The \emph{Wolbachia} transmission model (Figure \ref{Fig:diagram}) can be expressed as the system of ordinary differential equations
\allowdisplaybreaks
\begin{subequations} \label{E:mainmodel1}
\begin{align}
\frac{\dif A_{u}}{\dif t} &=B_{uu}  +v_u(B_{wu} + B_{ww})-\mu_aA_u-\psi A_u\\
\frac{\dif A_{w}}{\dif t} &=v_w ( B_{wu}+ B_{ww})-\mu_aA_w-\psi A_w\\
\frac{\dif F_{u}}{\dif t} &= b_f \psi A_u -\mu_{fu} F_u\\
\frac{\dif F_{w}}{\dif t} &= b_f  \psi A_w -\mu_{fw} F_w\\
\frac{\dif M_{u}}{\dif t} &= b_m  \psi A_u -\mu_{mu} M_u\\
\frac{\dif M_{w}}{\dif t} &= b_m \psi A_w-\mu_{mw} M_w ~~.
\end{align}
\end{subequations}
%\mbox{with the egg laying rates:}\\
\iffalse
\notes{Should we pull out the female mosquitoes, $F_*$, from the $B$ to make them 'rates' ? MH}
\notes{The unit of the left-hand side is number per time, so $B$ already represents rates. Ling}
\fi
Because the vertical transmission and birth rates, $B_{**}$, depend on the sex of the infected or uninfected parents,
the model included the four egg laying situations
\begin{subequations} \label{E:birthrates}
\begin{align}
B_{uu} &= \phi_{u} F_u m_u \left(1-\frac{N_A}{K_a}\right )\\
B_{uw} &= 0 \\
B_{wu} &= \phi_{w} F_w m_u \left(1-\frac{N_A}{K_a}\right)\\
B_{ww} &= \phi_{w} F_w m_w \left(1-\frac{N_A}{K_a}\right) ~~.
\end{align}
\end{subequations}
Here $m_u= \frac{M_u}{M_w+M_u} $ and  $m_w=  1 - m_u $  are the fractions of uninfected and infected male mosquitoes.
$B_{uu}$ is the egg laying rate of uninfected females mating with uninfected males. $B_{uw} $ is the egg laying rate of uninfected females mating with infected males. $B_{wu}$ is the egg laying rate of infected females mating with uninfected males. $N_A$ is the total number of aquatic stage of mosquitoes, and $K_a$ is carrying capacity of aquatic stage of mosquitoes.
These equations reflect the observations in Table \ref{table:Production of offspring} that
\begin {itemize}
\item{Mating of uninfected males   with uninfected females produce uninfected offspring.}
\item{Mating of infected males  with uninfected females leads to death of embryos  before hatching due to cytoplasmic incompatibility.}
\item{Uninfected males and  infected females produce a fraction, denoted by $v_w$, of  infected offspring by vertical transmission.}
\item{Cross of infected males  with infected females produces a   fraction of infected offspring. }
\end{itemize}

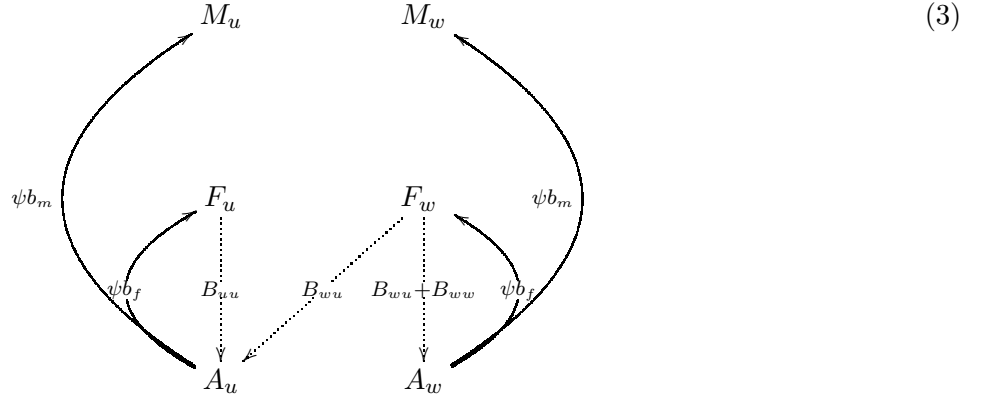
\begin{figure}
\centering
\begin{equation}
\xymatrix{
M_u
&&
M_w\\
&&\\
F_u
\ar@{.>}[dd]|-{B_{uu}}
&&
F_w
\ar@{.>}[ddll]|-{B_{wu}}\
\ar@{.>}[dd]|-{B_{wu}+B_{ww}}\\
&&\\
A_u
\ar@/^3pc/[uu]|{\psi  b_f  }
\ar@/^5pc/[uuuu]^{\psi b_m }
&&
A_w
\ar@/_3pc/[uu]|{\psi b_f  }
\ar@/_5pc/[uuuu]^{\psi b_m }
\\
&&
}
\end{equation}
\caption{
 The birthing rates (\ref{E:birthrates}) capture that when the uninfected males  mate with uninfected females,
 they produce uninfected offspring.
When infected males  mate with uninfected females, then CI causes the embryos to die before hatching. Uninfected males mating with  infected females produce a fraction, denoted by $v_w$, of  infected offspring by vertical transmission.
Cross of infected males  with infected females produces a   fraction of infected offspring.
}
\label{Fig:diagram}\end{figure}
\iffalse
\notes{Do we really need this plot?  I'm not sure that it adds much to the paper.  If we do keep it, then how about using dashed lines for egg laying? Otherwise the equations look disconnected.  MH}\notes{I prefer to keep this figure. Ling}
\fi

\begin{table}[h!]
\begin{center}
\begin{tabular}{lp{6.5in}}
$M_{u}$: &Number of susceptible male mosquitoes\\
$M_{w}$: &Number of  male mosquitoes  infected with  \emph{Wolbachia} \\
$F_{u}$: &Number of susceptible female mosquitoes\\
$F_{w}$: &Number of female mosquitoes  infected with  \emph{Wolbachia} \\
$A_{u}$: &Number of susceptible aquatic stage of mosquitoes\\
$A_{w}$:&Number of infected aquatic stage of mosquitoes\\
$N_A$:&Total number of  aquatic stage of mosquitoes\\
$K_a$: &Carrying capacity of aquatic stage of mosquitoes\\
$b_f$: & Fraction of births that are female mosquitoes. \\
$b_m$: & Fraction of births that are male mosquitoes $=1-b_f$. \\
$m_w$: & Fraction of the male mosquitoes that are infected $=M_w/(M_w+M_u)$. \\
$m_u$: & Fraction of the male mosquitoes that are uninfected $=1-m_w$. \\
$v_w$: & Fraction of infected mosquito eggs produced by infected  female mosquitoes.   \\
$v_u$: & Fraction of uninfected mosquito eggs produced by infected  female mosquitoes $=1-v_w$.   \\
$\phi_u$: & Per capita egg laying rate by \emph{Wolbachia}-free mosquito eggs. Number of eggs/time\\
$\phi_w$: & Per capita egg laying rate by \emph{Wolbachia}-infected mosquito eggs. Number of eggs/time\\
$\psi$: & Per capita development rate of mosquito eggs.  Time$^{-1}$\\
%$\eta$: &  Fraction of uninfected mosquito eggs  fertilized by infected male mosquitoes and can not survive. Dimensionless.\\
$\mu_a$: & Per capita death rate of  aquatic stage of mosquitoes. Time$^{-1}$  \\
$\mu_{fu}$: & Per capita death rate  of uninfected female mosquitoes. Time$^{-1}$  \\
$\mu_{fw}$: & Per capita death rate  of infected female mosquitoes. Time$^{-1}$  \\
$\mu_{mu}$: & Per capita death rate of uninfected male mosquitoes. Time$^{-1}$\\
$\mu_{mw}$: & Per capita death rate of infected male mosquitoes. Time$^{-1}$
%$\chi_{ss}$: & Mating rate between susceptible male mosquitoes and susceptible female mosquitoes.  $\chi_{ss}=\frac{S_M}{N_M}\frac{S_F}{N_F}\chi$.\\
%$\chi_{is}$: Mating rate between infected male mosquitoes and susceptible female mosquitoes. $\chi_{is}=\frac{I_M}{N_M}\frac{S_F}{N_F}\chi$.  \\
%$\chi_{si}$: Mating rate between susceptible male mosquitoes and infected female mosquitoes. $\chi_{si}=\frac{S_M}{N_M}\frac{I_F}{N_F}\chi$.   \\
%$\chi_{ii}$: Mating rate between infected male mosquitoes and infected female mosquitoes. $\chi_{ii}=\frac{I_M}{N_M}\frac{I_F}{N_F}\chi$.
\end{tabular}
\caption{State variables and parameters for the model
\eqref{E:mainmodel1}}\label{T:statevarandparameters}
\end{center}
\end{table}
\begin{table}
\centering
\begin{tabular}{|c|c|c|c|}
\hline
Male &Female& Offspring\\
\hline
 Uninfected &  Uninfected & Uninfected\\
\hline
 Uninfected&Infected & Some offspring are infected\\
\hline
Infected& Uninfected & Unviable\\
\hline
Infected&Infected & Some offspring are infected\\
\hline
\end{tabular}
\caption{The model allows for \emph{Wolbachia} to be transmitted vertically from
infected parents to their offspring.
The offspring of  male and female uninfected mosquitoes are infected.  Some of
the offspring of male and female infected mosquitoes are infected, as are the offspring
of an uninfected male and an infected female mosquitoes.  The offspring of an
infected male and uninfected female mosquito are unviable. }
\label{table:Production of offspring}
\end{table}
\subsection{Model Analysis}
We compute the basic reproduction number and  the equilibria, and analyze the stability of the  equilibrium points.

We use  the next generation matrix approach to compute the basic reproduction number \cite{Van2002}. Only infected compartments are considered for ease of computation:
\begin{align*}&\frac{d}{dt}\left[
\begin{array}{rlll}A_w &F_w &M_w \end{array}\right]^T= \mathscr{F}-\mathscr{V}=\\&\begin{bmatrix}
\left(v_w F_w\phi_w m_u +v_w F_w\phi_w m_w \right)\left(1-\frac{N_A}{K_a}\right)\\
 b_f  \psi A_w \\
 b_m \psi A_w
\end{bmatrix}-\begin{bmatrix}
\mu_aA_w+\psi A_w\\
\mu_{fw} F_w\\
\mu_{mw} M_w
\end{bmatrix},
\end{align*}
where $\mathscr{F}=(\mathscr{F}_i)$ is a vector for new infected, and  $\mathscr{V}=(\mathscr{V}_i)$ is a vector for transfer between compartments.

Jacobian matrices for transmission, $F$, and transition, $V$,  \cite{Van2002} are defined as:
\begin{equation}\label{Jacobian}
F= \left[\frac{\partial \mathscr{F}_i (x^0)}{\partial x_j}\right], \quad V= \left[\frac{\partial \mathscr{V}_i(x^0)}{\partial x_j}\right],
\end{equation}
where $x^0$ represents the disease free equilibrium, and $x_j$ is the number or proportion of infected individuals in compartment $j$, $j=1, 2, \cdots, m$.

The unique disease free equilibrium is
\begin{align*}
A_u^0&=K_a\left(1-\frac{1}{R_{0u}}\right)\\
F_u^0&= b_f  \frac{\psi}{\mu_{fu}}A_u^0 \\
M_u^0&=b_m \frac{\psi}{\mu_{mu}}A_u^0\\
A_w^0&=F_w^0=M_w^0=0~~,
\end{align*}
where $R_{0u}=\frac{ b_f \phi_u\psi}{(\mu_a+\psi)\mu_{fu}}$ is the threshold for \emph{Wolbachia}-free offspring, $\frac{b_f\phi_u}{\mu_{fu}}$ is the total number of eggs laid by uninfected female mosquitoes,  $\frac{\psi}{\mu_a+\psi}$ is the probability that aquatic stage of mosquitoes survive to the point when they develop into adult mosquitoes. $R_{0u}$ is the number of female eggs that develop into adult mosquitoes. When $R_{0u}>1$, then the mosquito population may grow; otherwise, the population will decrease.

The Jacobian matrix of $\mathscr{F}$ evaluated at DFE is:
\begin{align*}
F&=\begin{bmatrix}
0&v_w  m_u^0 \phi_w\left(1-\frac{N_A^0}{K_a}\right)& 0\\
 b_f \psi&0& 0\\
 b_m \psi&0&0
\end{bmatrix}=\begin{bmatrix}
0& \frac{v_w(\mu_a+\psi)\mu_{fu}\phi_w}{ b_f \psi\phi_u}& 0\\
 b_f \psi&0& 0\\
 b_m \psi&0&0
\end{bmatrix},\end{align*}
where $m_u^0=\frac{M_u^0}{M_w^0+M_u^0}$.

The Jacobian matrix of $\mathscr{V}$ evaluated at DFE is:
\begin{equation*}
V=\begin{bmatrix}
\mu_a+\psi &0& 0\\
0&\mu_{fw}&0\\
0&0& \mu_{mw}\\
\end{bmatrix}.
\end{equation*}
%the inverse of matrix $V$ is:
%\begin{equation*}
%V^{-1}=\begin{bmatrix}
%\frac{1}{\mu_a+\psi}&0& 0\\
%0&\frac{1}{\mu_{fw}}&0\\
%0&0& \frac{1}{\mu_{mw}}\\
%\end{bmatrix}. \end{equation*}

The next generation matrix for infected compartments at disease free equilibrium is:
\begin{equation*}
FV^{-1}=\begin{bmatrix}
0&\frac{v_w(\mu_a+\psi)\mu_{fu}\phi_w}{ b_f \psi\mu_{fw}\phi_u}& 0\\
\frac{ b_f \psi}{\mu_a+\psi }&0& 0\\
\frac{ b_m \psi}{\mu_a+\psi }&0&0
\end{bmatrix}.
\end{equation*}

The basic reproduction number is the largest abosulute eigenvalue of $FV^{-1}$, denoted by $\rho(FV^{-1})$:
\begin{equation}
R_0=\rho(FV^{-1})=\sqrt{v_w   \frac{\phi_w}{\mu_{fw}} \Big(\frac{\phi_u}{\mu_{fu}}\Big)^{-1}} ~~. \label{equation:R0}
\end{equation}

\emph{Wolbachia}-infected  and \emph{Wolbachia}-free and females  produce
  $\frac{\phi_w}{\mu_{fw}}$ and $\frac{\phi_u}{\mu_{fu}}$ eggs during their lifetime, respectively. Hence,
  $\frac{\phi_w}{\mu_{fw}} \Big(\frac{\phi_u}{\mu_{fu}}\Big)^{-1}$
  is the ratio of the number of eggs produced by \emph{Wolbachia}-infected females to the number of  eggs produced by \emph{Wolbachia}-free females during their lifetime.  $R_0$ is geometric mean of  vertical transmission rate and  the ratio of the total number of eggs produced by \emph{Wolbachia}-infected females to the total number of  eggs produced by \emph{Wolbachia}-free females.
  %Therefore, invasion of  \emph{Wolbachia}  in natural
%populations depends upon fraction of vertical transmission and the number of eggs produced by  \emph{Wolbachia}-infected  \emph{Aedes aegypti}  relative to those by uninfected competitors.
%Local stability of Completed model at Disease free equilibrium

When the variables are ordered as: $[A_u \ F_u \ M_u \  A_w \  F_w \ M_w ]^T$, Jacobian matrix at equilibrium $[A_u^* \ F_u^* \ M_u^* \  A_w^* \  F_w^* \ M_w^* ]^T$ for system  of equations $(\ref{E:mainmodel1})$ is:
%\begin{align*}
\[
J=\begin{bmatrix}
c&a&b&|&c_2 &a_2 &b_2\\
 b_f \psi&-\mu_{fu}&0&|&0 &0 &0\\
 b_m \psi&0& -\mu_{mu}&|&0 &0 &0\\
--&--&--&--&--&--&--\\
c_3&0&b_3&|&c_1&a_1&b_1\\
0&0&0&|& b_f \psi&-\mu_{fw}&0\\
0&0&0&|& b_m \psi&0& -\mu_{mw}\\
\end{bmatrix}
=\begin{bmatrix}
A&C\\
D&B\\
\end{bmatrix},
\]
%\end{align*}
where
\allowdisplaybreaks
\begin{align*}
a&=\phi_u m_u^* \left(1-\frac{A_u^*+A_w^*}{K_a}\right)\\
b%&=\frac{M_w^*}{(M_w^*+M_u^*)^2}(1-\frac{A_u^*+A_w^*}{K_a})[F_u^*\phi_u+v_uF_w^*\phi_w]- (1-\frac{A_u^*+A_w^*}{K_a})v_uF_w\phi_w\frac{M_w^*}{(M_w^*+M_u^*)^2}\\
&=\phi_u m_w^* \frac{F_u^*}{M_w^*+M_u^*}\left(1-\frac{A_u^*+A_w^*}{K_a}\right)\\
c%&=-(\mu_a+\psi)-[F_u^*\phi_u+v_u F_w^*\phi_w]\frac{M_u^*}{(M_u^*+M_w^*)K_a}-v_u F_w^*\phi_w\frac{M_w^*}{(M_u^*+M_w^*)K_a} \\
&=-(\mu_a+\psi)-\phi_u m_u^*\frac{F_u^*}{K_a}-\phi_wv_u\frac{F_w^*}{K_a}\\
a_1%&=(\frac{v_w\phi_w^*M_u^*+v_w \phi_wM_w^*}{M_u^*+M_w^*})\left(1-\frac{A_u^*+A_w^*}{K_a}\right)\\
&=v_w\phi_w\left(1-\frac{A_u^*+A_w^*}{K_a}\right)\\
b_1%&=\frac{-v_w F_w^*\phi_wM_u^*+v_w F_w^*\phi_wM_u^*}{(M_w^*+M_u^*)^2}(1-\frac{A_u^*+A_w^*}{K_a})\\
&=0\\
c_1%&=-(\mu_a+\psi)-\frac{v_w\phi_wF_w^*M_u^*}{K_a(M_u^*+M_w^*)}-\frac{v_w\phi_w^*F_w^*M_w^*}{K_a(M_u^*+M_w^*)}\\
&=-(\mu_a+\psi)-\phi_w v_w  \frac{F_w^*}{K_a}\\
a_2&=\phi_w v_u \left(1-\frac{A_u^*+A_w^*}{K_a}\right)\\
b_2%&=-[F_u^*\phi_u+v_u\phi_wF_w^*]\frac{M_u^*}{(M_w^*+M_u^*)^2}(1-\frac{A_u^*+A_w^*}{K_a})+v_u\phi_wF_w^*]\frac{M_u^*}{(M_w^*+M_u^*)^2}(1-\frac{A_u^*+A_w^*}{K_a})\\
&=-\phi_u m_u^* \frac{F_u^*}{M_w^*+M_u^*}\left(1-\frac{A_u^*+A_w^*}{K_a}\right)\\
c_2%&=-[F_u^*\phi_u+v_uF_w^*\phi_w]\frac{M_u^*}{(M_w^*+M_u^*)K_a}-v_u F_w^*\phi_w\frac{M_w^*}{(M_u^*+M_w^*)K_a} \\
&=- \phi_u m_u^* \frac{F_u^*}{K_a}-\phi_wv_u \frac{F_w^*}{K_a}\\
b_3%&=\phi_wF_w^*\frac{(v_w M_w^*-v_w M_w^*)}{(M_w^*+M_u^*)^2}(1-\frac{A_u^*+A_w^*}{K_a})\\
&=0\\
c_3%&=-[v_w M_u^*+v_w M_w^*]\frac{F_w^*\phi_w}{(M_u^*+M_w^*)K_a}\\
&=-\phi_w v_w \frac{ F_w^*}{K_a}
\end{align*}
At disease free equilibrium,
$D=0$, then the eigenvalues of $J$ are eigenvalues of $A$ and $B$. Matrices $A$ and $B$ are:
%\begin{equation*}
%A=\begin{bmatrix}
%-(\mu_a+\psi) &\frac{\mu_{fu}(\mu_a+\psi)}{ b_f \psi}& \frac{\mu_{mu}(\mu_a+\psi)}{ b_m \phi\psi}\\
% b_f \psi&-\mu_{fu}&0\\
% b_m \psi&0& -\mu_{mu}\\
%\end{bmatrix}\end{equation*}
\begin{equation*}
A=\begin{bmatrix}
-(\mu_a+\psi)-\phi_u \frac{F_u^0}{K_a} &\frac{\mu_{fu}(\mu_a+\psi)}{ b_f \psi}& 0\\
 b_f \psi&-\mu_{fu}&0\\
 b_m \psi&0& -\mu_{mu}\\
\end{bmatrix}, \ \ \
B=\begin{bmatrix}
-(\mu_a+\psi) &\frac{v_w(\mu_a+\psi)\mu_{fu}\phi_w}{ b_f \psi\phi_u}& 0\\
 b_f \psi&-\mu_{fw}&0\\
 b_m \psi&0& -\mu_{mw}\end{bmatrix}.  \end{equation*}

If $R_{0u}>1$, then all eigenvalues of $A$ are negative.  If $R_0<1$, then all eigenvalues of $B$ are negative. Therefore, the system $(\ref{E:mainmodel1})$ at disease free equilibrium is LAS whenever $R_{0u}>1$ and $R_0<1$.
\subsubsection{Complete Vertical Transmission}
If $R_0<1$ and $R_{0w}>1$ and the vertical transmission is $100\%$, ($v_w=1$), then
the ratio of the infected to the uninfected aquatic stage mosquitoes, $k$,  is
\begin{equation*}
k = \frac{A_w^*}{A_u^*}=
%\frac{ F_w^*\phi_w}{F_u^*\phi_u } \frac{ 1}{m_u^* }=
\frac{\mu_{mw}}{\mu_{mu}}( {R_0^{-2}-1})~~,
\end{equation*}
and there is a unique endemic equilibrium:
\allowdisplaybreaks
\begin{subequations} \label{endemic equilibrium}
\begin{align}
A_{u}^*
&=\frac{ K_a}{1+k}\left(1-{R_{0w}^{-1}}\right)\\
A_{w}^*
&=k A_{u}^* \\
F_{u}^*
&=b_f   \frac{ \psi }{\mu_{fu}} A_{u}^* \\
F_{w}^*
&=k F_{u}^*\\
M_{u}^*
&= b_m   \frac{ \psi }{\mu_{mu}} A_{u}^*\\
M_{w^*}&=k M_{u}^*,
%k&=\frac{\phi_w\mu_{mu}\mu_{fu}}{\mu_{mw}(\phi_u\mu_{fw}-\phi_w\mu_{fu})}\\
%\frac{k}{k+1}&=\frac{\phi_w\mu_{mu}\mu_{fu}}{\phi_w\mu_{mu}\mu_{fu}+\mu_{mw}\phi_u\mu_{fw}-\phi_w\mu_{mw}\mu_{fu}}
\end{align}
\end{subequations}
where $R_{0w}=\frac{v_w b_f \psi\phi_w}{(\mu_a+\psi)\mu_{fw}}$.

%We apply Gershgorin circle theorem to find the bounds of eigenvalues.

%First row,
%$|\lambda_1-c|\leq a+b-c_2\Rightarrow \lambda \leq \frac{(\mu_a+\psi)\mu_{fu}}{ b_f \psi}+\frac{(\mu_a+\psi)\mu_{fu}\mu_{mu}\phi_w}{ b_m \psi\phi_u\mu_{fw}}-\mu_a-\psi$.

%Second row,
%$|\lambda_2+\mu_{fu}|\leq  b_f  \psi \Rightarrow \lambda \leq  b_f  \psi - \mu_{fu}$.

%Third row,
%$|\lambda_3+\mu_{mu}|\leq  b_m  \psi \Rightarrow \lambda \leq   b_m  \psi - \mu_{mu}$.

%Fourth row,
%$|\lambda_4-c_1|\leq a_1+|c_3|\Rightarrow \lambda \leq \frac{(\mu_a+\psi)\mu_{fu}}{ b_f \psi}-(\mu_a+\psi)$.

%Fifth row,
%$|\lambda_5+\mu_{fw}|\leq  b_f  \psi \Rightarrow \lambda \leq  b_f  \psi - \mu_{fw}$.

%Sixth row,
%$|\lambda_6+\mu_{mw}|\leq  b_m  \psi \Rightarrow \lambda \leq  b_m  \psi - \mu_{mw}$.

%Numerical simulations show that there are always positive eigenvalues.
%The reproduction number for  \emph{Wolbachia}  infected offspring is:
%We define $\frac{ b_f \phi_w\psi}{(\mu_a+\psi)\mu_{fw}}$ as the reproduction number for  \emph{Wolbachia}  infected offspring, denoted by $R_{0w}$.
If we further assume that $\mu_{mu}=\mu_{mw}$, then $k= R_0^{-2}-1$
% I don't think repeating the equations adds much to the paper. MH
\iffalse
and the endemic equilibrium is:
\begin{align*}
A_{u}^*&=R_0^2K_a\left(1- R_{0w}^{-1} \right)\\
A_{w}^*&= k A_{u}^* \\
F_{u}^*&= b_f  \frac{\psi  }{\mu_{fu}}A_{u}^*   \\
F_{w}^*&= b_f  \frac{\psi  }{\mu_{fw}}A_{w}^*  \\
M_{u}^*&= b_m  \frac{\psi  }{\mu_{mu}} A_{u}^*\\
M_{w^*}&= b_m  \frac{\psi  }{\mu_{mw}} A_{w}^*  ~~.
\end{align*}
%If $R_0^2$ increases, then $A_u^*$ increases.\\
\fi
and the unique complete  infection equilibrium (CIE) is:
\begin{align*}
A_u&=0\\
A_w&=K_a\left(1- R_{0w}^{-1} \right)\\
F_u&=0\\
F_w&=b_f \frac{  \psi }{\mu_{fw}}A_w \\
M_u&=0\\
M_w&= b_m  \frac{\psi }{\mu_{mw}} A_w .
\end{align*}
Jacobian matrix of system of equations $(\ref{E:mainmodel1})$ at complete  \emph{Wolbachia}  infection equilibrium for  $v_w=1$  is:
\begin{align*}
J_{cw}&=\begin{bmatrix}
-\mu_a-\psi&0&0&|&0 &0 &0\\
 b_f \psi&-\mu_{fu}&0&|&0 &0 &0\\
 b_m \psi&0&-\mu_{mu}&|&0 &0 &0\\
--&--&--&--&--&--&--\\
-\frac{ b_f  \psi (1- R_0^2)\phi_w}{\mu_{fw}}\left(1- R_{0w}^{-1} \right)&0&0&|&-\frac{\phi_w b_f \psi}{\mu_{fw}}&\frac{(\mu_a+\psi)\mu_{fw}}{ b_f \psi}&0\\
0&0&0&|& b_f \psi&-\mu_{fw}&0\\
0&0&0&|& b_m \psi&0& -\mu_{mw}\\
\end{bmatrix}
=\begin{bmatrix}
 A&0\\
C&B\\
\end{bmatrix}.\end{align*}
Eigenvalues of $J_{cw}$ are composed of three eigenvalues of $A$ and three eigenvalues of $B$. The eigenvalues of $A$ are all negative.
Characteristic polynomial of $B$ is:
\begin{equation}
\label{char_eqn}
\lambda^2+\left(\mu_{fw}+\frac{\phi_w b_f \psi}{\mu_{fw}}\right)\lambda+\phi_w b_f \psi-(\mu_a+\psi)\mu_{fw}=0~~.
\end{equation}
If $R_{0w}>1$,  all eigenvalues of $B$ are negative. Therefore, the complete infection equilibrium is LAS whenever $R_{0w}>1$, as shown in Figure \ref{fig:nu=1endemic}.

\subsubsection{Incomplete Vertical Transmission}
When vertical transmission is incomplete, i.e., $0<v_w<1$, then at endemic equilibrium,
\begin{equation}
\frac{A_w^*}{A_u^*}=\frac{v_w-v_u \frac{\mu_{mw}}{\mu_{mu}}\pm
\sqrt{\left(v_u\frac{\mu_{mw}}{\mu_{mu}}-v_w\right)^2
%-4v_u\left(\frac{v_w\mu_{mw}}{R_0^2\mu_{mu}}-\frac{v_w\mu_{mw}}{\mu_{mu}}\right)}}.$$
-4v_u v_w (R_0^{-2}-1) \frac{\mu_{mw}}{\mu_{mu}}  }} {2v_u}~~.\end{equation}
We assume $\mu_{mu}=\mu_{mw}$, and
let $ k=\frac{A_w^*}{A_u^*}$. When
$\left(v_u \frac{\mu_{mw}}{\mu_{mu}}-v_w\right)^2 > 4v_u v_w (R_0^{-2}-1) \frac{\mu_{mw}}{\mu_{mu}}$, that is,
$R_0 >2\sqrt{ v_wv_u}$.  Equation \ref{char_eqn} has two roots:
$k_1=\frac{2v_w-1+\sqrt{1-4v_wv_uR_0^{-2}}}{2v_u}$ and  $k_2=\frac{2v_w-1-\sqrt{1-4v_wv_uR_0^{-2}}}{2v_u}$.

%Since $4v_wv_u-1=-(2v_w-1)^2\leq 0$,  $4v_wv_u \leq 1$.Hence, $R_0^2\geq 4v_wv_u$ whenever $R_0>1$.
%We plug in the values of $k$ into endemic equilibria in Equations (\ref{endemic equilibrium}).

If $0.5<v_w<1$ and $4v_w(v_w-1)<R_0<1$,  two endemic equilibria exist as shown in Figure \ref{fig:nu=0.9endemic}, \ref{fig:nu=0.8endemic}, and \ref{fig:nu=0.75endemic}.
When $k=k_1$, the endemic equilibrium is  LAS. When  $k=k_2$, the equilibrium is not LAS and backward bifurcation occurs.
If $v_w\leq 0.5$ and $R_0\leq 1$, then endemic equilibrium does not exist, only DFE exists. When $k_1=k_2$, then $R_0=2\sqrt{v_wv_u}$. Let $R_0^*=2\sqrt{v_wv_u}$, which is the intersection of unstable and stable endemic equilibrium.
 When $R_0> 1$, a unique endemic equilibrium point exists with $k=k_1$. It is proven to  be LAS by numerical simulations.

%If $v_w>0.5$  and $R_0=1$, only  one endemic equilibrium exists (two equilibrium are equal to each other.

%If $v_w=0.5$ and $R_0\leq 4v_w(v_w-1)=1$, then endemic equilibrium does not exist.

%If $v_w<0.5$ and $R_0\leq 1$, then endemic equilibrium does not exist.

\begin{figure}[!htbp]
\centering
\subfigure[$v_w=1$.]{
\label{fig:nu=1endemic}
\includegraphics[angle=0,width=6.7cm,height=5.8cm]{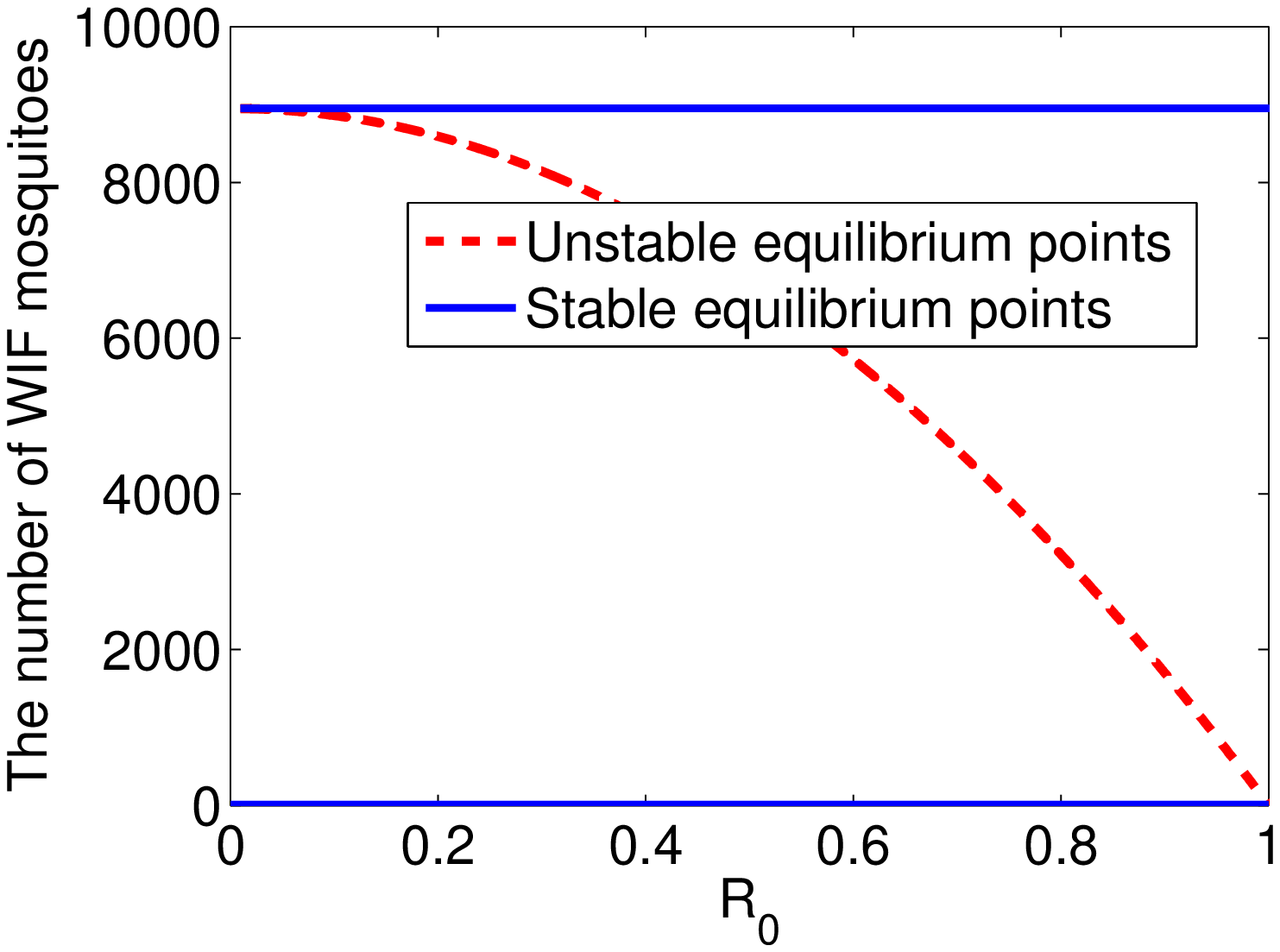}}
\hspace{0.1in}
\subfigure[$v_w=0.9$.]{
\label{fig:nu=0.9endemic}
\includegraphics[angle=0,width=6.7cm,height=5.8cm]{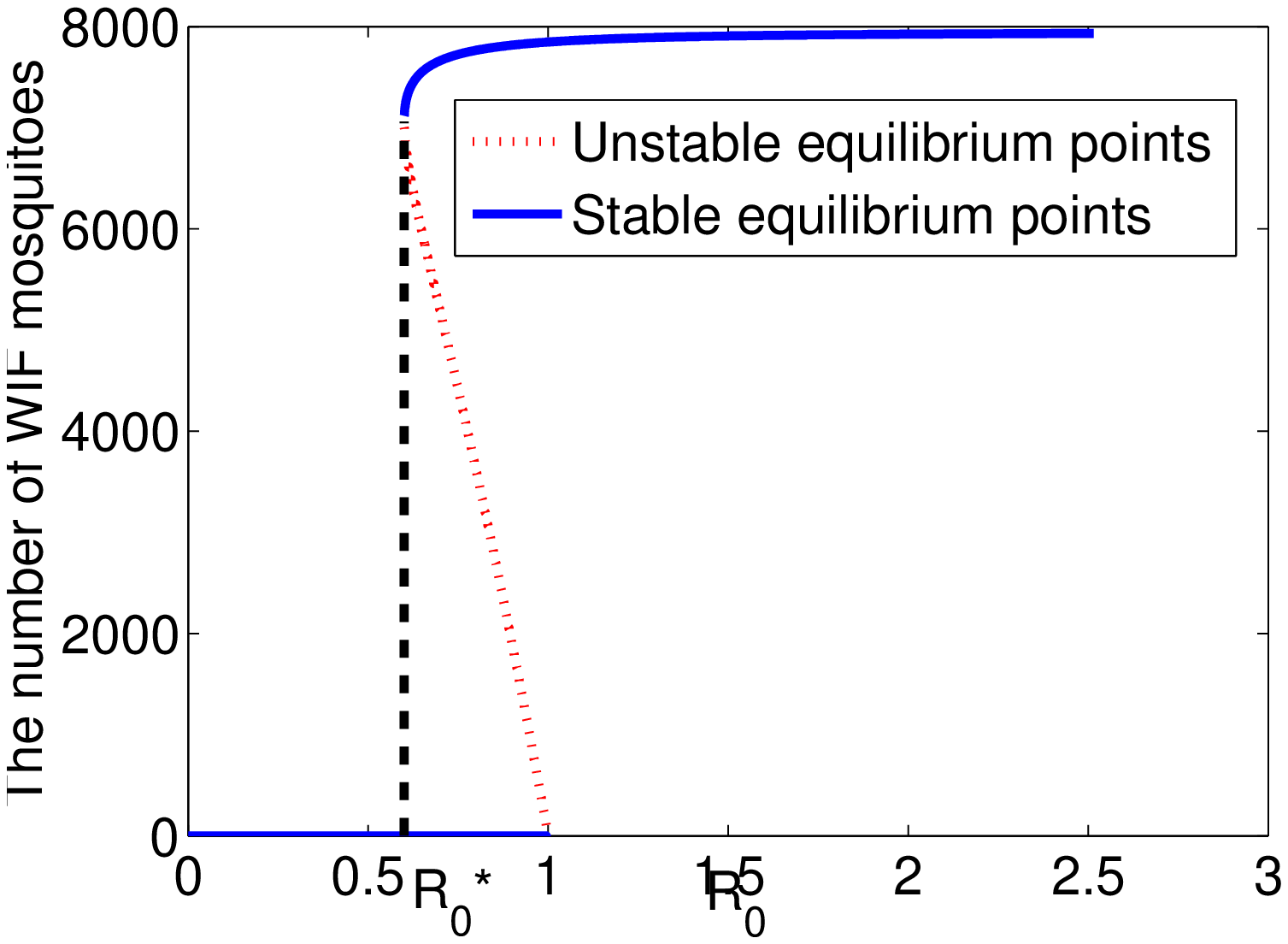}}
\hspace{0.1in}
\subfigure[ $v_w=0.8$.]{
\label{fig:nu=0.8endemic}
\includegraphics[angle=0,width=6.7cm,height=5.8cm]{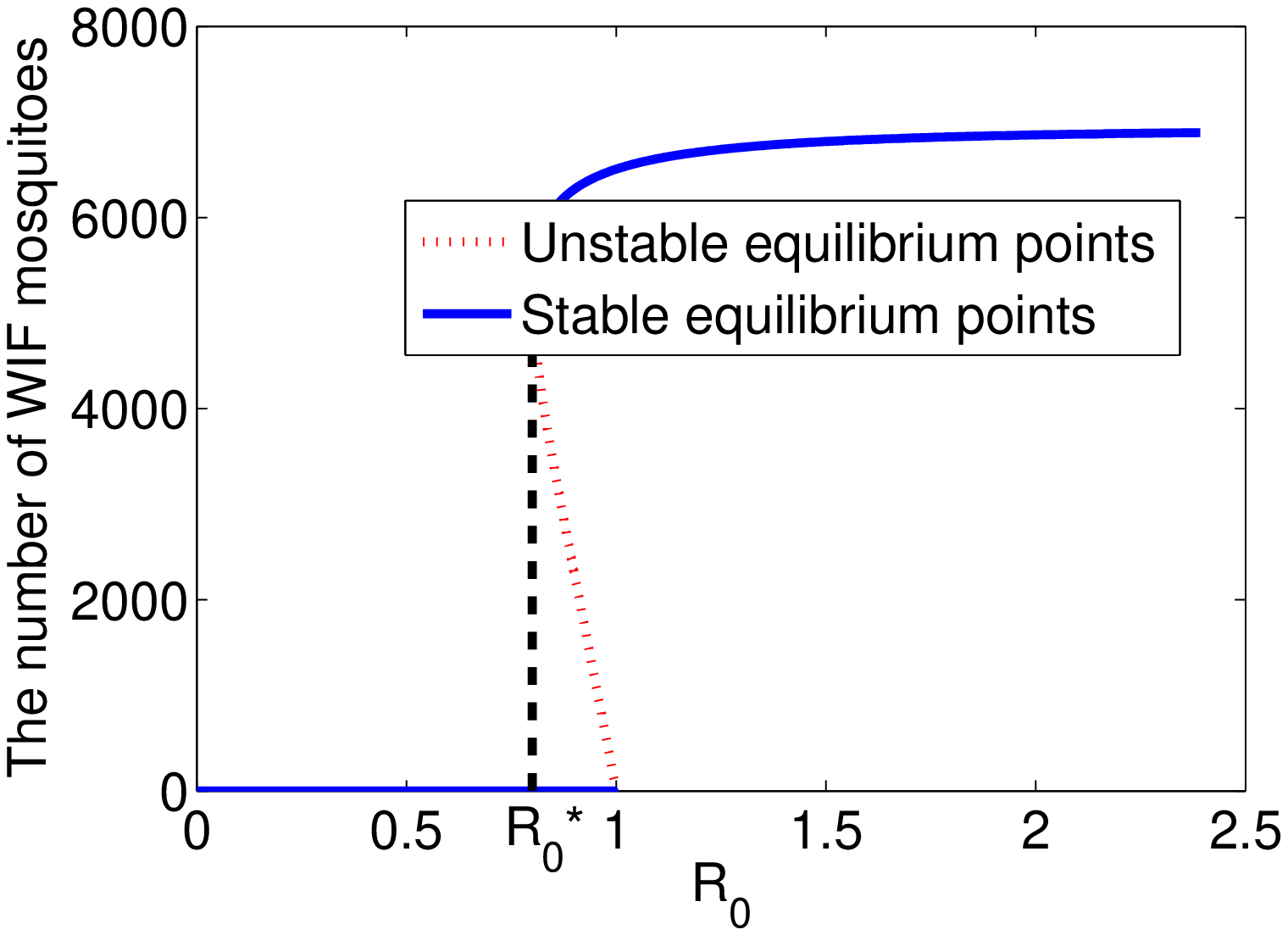}}
\hspace{0.1in}
\subfigure[ $v_w=0.75$.]{
\label{fig:nu=0.75endemic}
\includegraphics[angle=0,width=6.7cm,height=5.8cm]{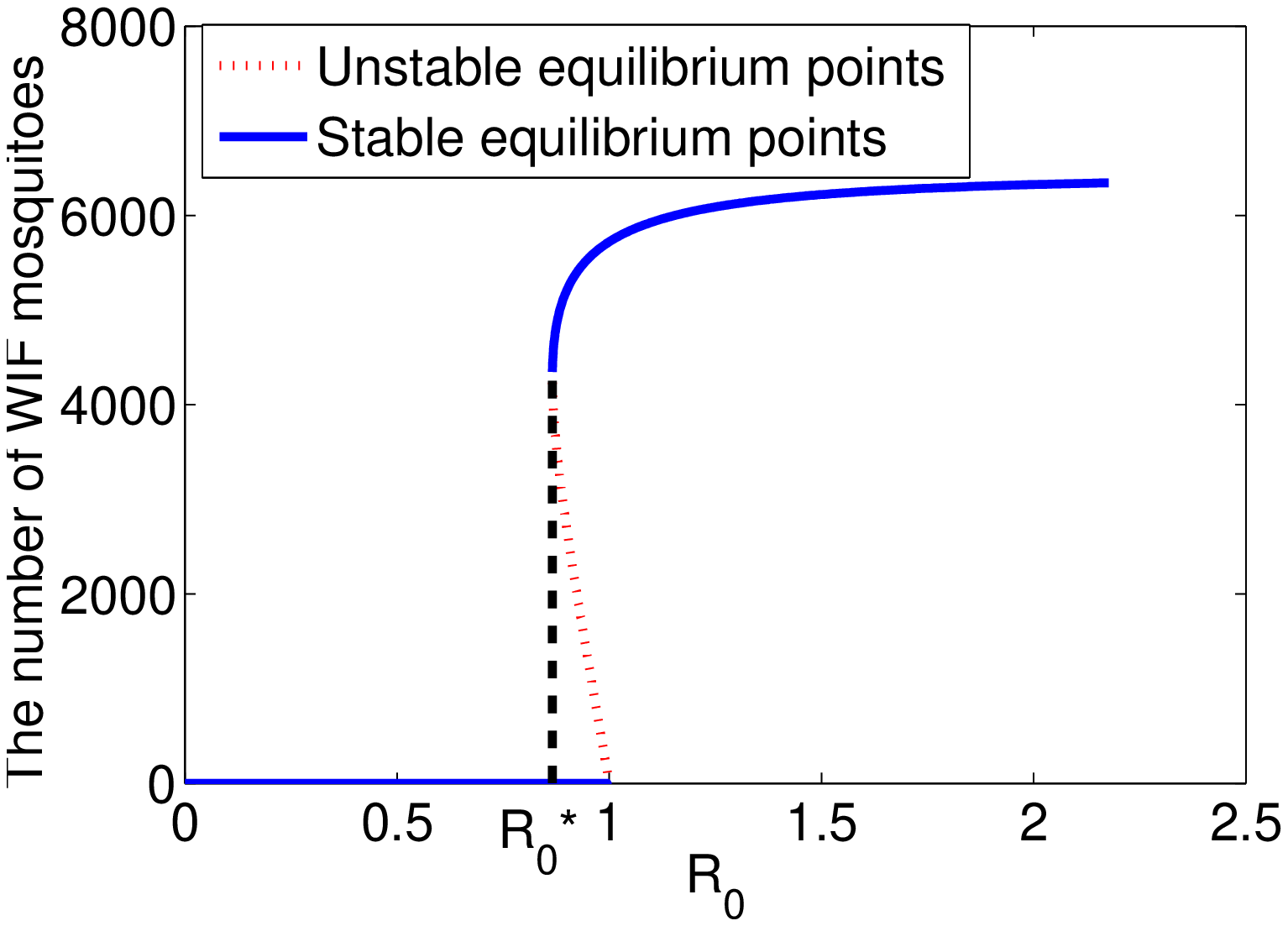}}
\hspace{0.1in}
\caption{Bifurcation diagrams for  \emph{Wolbachia}  vertical transmission.  $\phi_u$ and $\mu_{fu}$ are varying, other parameter values are the same as those baseline values in Table \ref{table:baseline and ranges}. Denote the intersection of two endemic equilibrium, that is, the intersection of the black dashed line and  x-axis, as $R_0^*=\sqrt{4v_wv_u}$. When   $R_0<1$ and $v_w<0.5$, no endemic equilibria exist.
When $R_0<1$ and $v_w>0.5$, as the vertical transmission rate increases so does $R_0^*$ and the LAS equilibrium approaches a constant.   If we increase the number  of infected females, then the endemic equilibrium may become stable endemic or complete infection equilibrium. If we decrease the number of infected females at endemic equilibrium, then the endemic equilibrium may become disease free equilibrium. WIF denotes \emph{Wolbachia}-infected female mosquitoes.}
\label{fig:bifurcation diagram}
\end{figure}

\begin{table}
\centering
\begin{tabular}{|c|c|c|c|}
\hline
Parameter &Baseline  Value & Ranges&References\\
\hline
$ b_f $ & $0.5$ & $0.34-0.6$&\cite{Lounibos2008, Delatte2009, Monteiro2007}\\
\hline
$\phi_u$& $50/$day & $0-75$&\cite{Delatte2009}\\
\hline
$\phi_w$& $51$/day & $0-75$&\cite{Delatte2009}\\
\hline
%$\psi$ & $1/7$ & $1/15-1/6$&\cite{Marquardt2005, Delatte2009}\\
%\hline
%$\psi$ & $0.0141$ & &Assume\\
%\hline
$\psi$ & $0.01/$day & &Assume\\
\hline
$\mu_a$ & $0.02/$day & &Assume\\
\hline
%$\eta$& $0.99$ & $0-1$&\cite{Walker2011, Dobson2004}\\
%\hline
$v_w$& $0.9$ & $0-1$&Assume\\
\hline
$\mu_{fu}$ &  $0.061/$day & $1/55-1/11$&\cite{Delatte2009, Styer2007, Marquardt2005}\\
\hline
$\mu_{fw}$ &  $0.068/$day & $1/55-1/11$&\cite{Delatte2009, Styer2007, Marquardt2005}\\
\hline
$\mu_{mu}$ &  $0.068/$day &  $1/31-/7$&\cite{Delatte2009, Styer2007, Marquardt2005}\\
\hline
$\mu_{mw}$ &  $0.068/$day & $1/31-/7$&\cite{Delatte2009, Styer2007, Marquardt2005}\\
\hline
%$\mu_a$& $1/9$ & $1/15-1/6$&\cite{Delatte2009, Marquardt2005}\\
%\hline
%$\mu_a$& $0.0363$ & &Assume\\
%\hline
%$\mu_e$ &$0.0100503$  &  & \cite{Luz2009, Rasgo2004, Styer2007, Williams2013}\\
%\hline
%$\mu_l$ &$0.10536$  &  & \cite{Barrera2006,Braks2006,Focks1993, Rasgo2004, Southwood1972,Styer2007a,Styer2007}\\
%\hline
%$\mu_p$  &$0.1$  &  &Assume based on \cite{Rueda1990}\\
%\hline
%$\psi_e$ &$0.06$  &  &Assume\\
%\hline
%$\psi_l$ &$0.15$  &  & \cite{Rueda1990}\\
%\hline
%$\psi_p$ &$1/8$  &  &\cite{Rueda1990}\\
%\hline
%$K_e$& $100000$ &&Assume\\
%\hline
$K_a$& $100,000$ &&Assume\\
\hline
\end{tabular}
\caption{Baseline values for parameters in Model (\ref{E:mainmodel1}).}
\label{table:baseline and ranges}
\end{table}

\begin{table}
\centering
\begin{tabular}{|p{2.4cm}|p{3.7cm}|p{5.5cm}|p{4cm}|}
\hline
Vertical  &Disease Free&Endemic &Complete Infection\\
Transmission  &Equilibrium (DFE)& Equilibrium (EE) &Equilibrium (CIE) \\
\hline
Complete &$R_0<1$ and $R_{0u}>1$, LAS&$R_0<1$ and $R_{0w}>1$, unstable&$R_{0w}>1$, LAS\\
\hline
 Incomplete &$R_0<1$ and $R_{0u}>1$, LAS&$0.5<v_w<1$ and $4v_w(v_w-1)<R_0<1$, when $k=k_1$, LAS, when $k=k_2$, unstable&does not exist\\
\hline
\end{tabular}
\caption{Threshold condition for existence of disease free equilibrium, endemic equilibrium, and complete infection equilibrium and their stability. $R_{0u}$ is the threshold for \emph{Wolbachia}-free mosquito population, and $R_{0w}$ is the threshold for \emph{Wolbachia}-infected mosquito population. Only when $R_{0u}>1$, \emph{Wolbachia}-free may grow, and only when $R_{0w}>1$, \emph{Wolbachia}-infected population may  grow.}
\label{table:summary_stability}
\end{table}

\section{Results}  \label{Results}
We observed that  \emph{Wolbachia}  can persist when $R_0^*<R_0$, where $R_0^*$
is the turning point of the backward bifurcation.
%When $R_0>1$, the number of  \emph{Wolbachia}  infected females.
Three equilibria, namely, disease free equilibrium, endemic equilibrium, and complete infection equilibrium coexist  when $R_0<1$,  $R_{0w}>1$, $v_w=1$, and $R_{0u}>1$ as shown in Figure  \ref{fig:nu=1endemic}.   The disease free equilibrium is LAS whenever  $R_0<1$ and $R_{0u}>1$, and    complete infection equilibrium exists and is LAS as long as  $R_{0w}>1$.  The unique  endemic equilibrium is not  LAS, which  can become disease free equilibrium by decreasing the number of infected individuals or become complete infection equilibrium by increasing the number of infected individuals.
Figures \ref{fig:nu=0.9endemic},  \ref{fig:nu=0.8endemic}, and  \ref{fig:nu=0.75endemic} showed that two endemic equilibrium points exist when $\sqrt{4v_wv_u}<R_0<1$ and $v_w>0.5$, and only one of them is LAS proven by numerical simulations.  When $v_w$ is larger,  $R^*$  is  closer to one. The condition for the existence of the equilibria and their stability are summarized in Table \ref{table:summary_stability}.

Initial condition thresholds for an epidemic to occur vary with the vertical transmission rate as shown in Table \ref{table:different initial condition threshold nu}. When $v_w=0.9$, the epidemic will spread if at least $ 40\%$ of the population are initially infected.  When $v_w=0.95$, the epidemic will spread if at least $34\%$ of the population are initially infected. When $v_w=1$, the epidemic will spread if at least $ 28\%$ of the population are initially infected.  When the vertical transmission rate is high, the threshold for the number of initially infected individuals to start  \emph{Wolbachia}  epidemic is low.

The initial condition threshold for an epidemic to occur  varies with the ratio of death rates of \emph{Wolbachia}-infected male mosquitoes to  death rates of \emph{Wolbachia}-free male mosquitoes, $\frac{\mu_{mw}}{\mu_{mu}}$, while fixing other parameters. The larger $\frac{\mu_{mw}}{\mu_{mu}}$ is, the larger the  initial number of infected individuals  required to start  a \emph{Wolbachia}  epidemic. With higher vertical transmission rate, a smaller percentage of  \emph{Wolbachia}  carriers can invade, as shown in Table \ref{table:different initial condition threshold nu}.
Similarly, if we increase $\phi_w$ or $\mu_{fu}$, or decrease $\phi_u$ or $\mu_{fw}$, then the threshold for initial infection is smaller. Increasing  $\frac{\psi}{\mu_a+\psi}$ will increase the epidemic threshold for initial infection.   The  thresholds for fraction of  initial infections decreasing with the basic reproduction number is shown in Figure \ref{fig:fraction}.

%integrateWolmos20.m
\begin{figure}
\begin{center}
\includegraphics[width=0.45\textwidth]{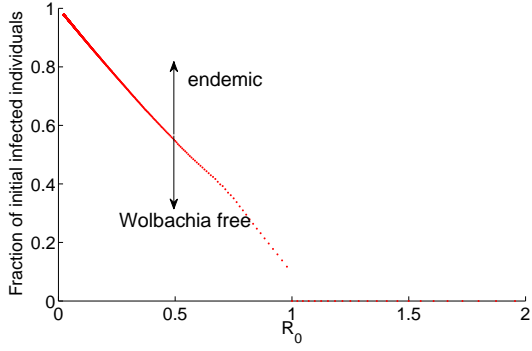}
\end{center}
\caption{Thresholds for fraction of infected individuals vary with reproduction number. $A_{u0}+A_{w0}=A_{u}^0$, $F_{u0}+F_{w0}=F_{u}^0$, $M_{u0}+M_{w0}=M_{u}^0$.  When $R_0<1$, the smaller $R_0$ is, the larger number of infected female mosquitoes are needed to be released for \it Wolbachia \rm to be endemic. The \it Wolbachia \rm infection is only sustained if the fraction of WIF mosquitoes is above the red dotted line.
\iffalse \notes{This figure needs additional explanation, like you used on the poster.  Say add arrows going up and down from the threshold conditions and a better description in the caption, like we have in figure 2. If we redraw the figure, then we should make the dots slightly bigger.   Since it costs the same amount of space to have two side-by-side figures as it does to have a single   figure, is there a companion figure we could use?  (I can't think of one.) MH}\fi
}
\label{fig:fraction}
\end{figure}

The reproduction number is  very sensitive to the vertical transmission rate, egg laying rates of \emph{Wolbachia}-infected mosquitoes,  egg laying rates of \emph{Wolbachia}-free mosquitoes,  and death rates of \emph{Wolbachia}-infected female mosquitoes and \emph{Wolbachia}-free female mosquitoes as shown in Equation (\ref{equation:R0}). The reproduction number varies directly with either the vertical transmission rate, the egg laying rates of \emph{Wolbachia}-infected mosquitoes, or the death rates of \emph{Wolbachia}-free female mosquitoes. The  reproduction number varies inversely with egg laying rates of \emph{Wolbachia}-free mosquitoes, or the death rates of \emph{Wolbachia}-infected female mosquitoes, but towards opposite direction.

We compared five strategies  before the release of \emph{Wolbachia}-infected female mosquitoes:
\begin{itemize}
\item DFE: releasing \emph{Wolbachia}-infected female mosquitoes at the disease free equilibrium,
\item KHA: first killing half of the aquatic stage of mosquitoes,
\item  KHM: first killing half of the wild adult mosquitoes,
\item  KHM2: first killing half of the wild adult mosquitoes, and then killing half of the adult mosquitoes again after two weeks,
\item  KHMA: first killing half of the wild mosquitoes and half aquatic stage of mosquitoes.
\end{itemize}
 The ratios of the minimum  number of \emph{Wolbachia}-infected female mosquitoes that can lead to persistence of \emph{Wolbachia} to the number of female mosquitoes at disease free equilibrium are listed in Table \ref{table:thresholds} in decreasing order.  Notice that killing the adult mosquitoes only once is not an effective strategy because the aquatic stage mosquitoes hatch and quickly replace the wild uninfected population.  Killing both the adult and aquatic (larvae) stage mosquitoes before releasing the infected mosquitoes is the most effective strategy.

\begin{table}
\centering
\begin{tabular}{|c|c|c|c|}
 \hline
Initial condition& $v_w=0.9$&$v_w=0.95$ &$v_w=1$\\
\hline
$A_{u0}$&$\leq 60 \%A_u^0$&$\leq 66 \%A_u^0$ &$\leq 72 \%A_u^0$\\
\hline
$A_{w0}$&$\geq 40\%A_u^0$&$\geq 34 \%A_u^0$ &$\geq 28 \%A_u^0$\\
\hline
$F_{u0}$&$\leq 60\%F_u^0$&$\leq 66 \%F_u^0$ &$\leq 72 \%F_u^0$\\
\hline
$F_{w0}$&$\geq 40 \%F_u^0$&$\geq 34 \%F_u^0$ &$\geq 28 \%F_u^0$\\
\hline
$M_{u0}$&$\leq 60 \%M_u^0$&$\leq 66 \%M_u^0$ &$\leq 72 \%M_u^0$\\
\hline
$M_{w0}$&$\geq 40 \%M_u^0$&$\geq 34 \%M_u^0$ &$\geq 28 \%M_u^0$\\
\hline
\end{tabular}
\caption{ Initial condition  thresholds for epidemic to occur with different vertical transmission rates. When the number is \emph{Wolbachia}-carrying mosquitoes is above the threshold, \emph{Wolbachia}-infected mosquitoes can petsit, otherwise, they are wiped out by \emph{Wolbachia}-free mosquitoes.\iffalse $\phi_u=4$,
$\phi_w=5$,
$\theta=0.5$,
$\mu_a=0.0204$,
$\mu_{fu}=0.061$,
$\mu_{fw}=0.068$,
$\mu_{mu}=0.068$,
$\mu_{mw}=0.068$,
$\psi=0.0067$, and
$K_A=204160$. \fi
\iffalse \notes{We could have a more detailed discussion of what we want the reader to 'take away' after reading this table.  MH}\fi
}
\label{table:different initial condition threshold nu}
\end{table}
%run integrateWolmos4.m again

%integrateWolmos25.m
\begin{table}
\centering
\begin{tabular}{|p{1.6cm}|p{3.6cm}|p{3.6cm}|p{3.6cm}|}
\hline
Scenario &minimum release ratio  when $R_0=0.85$ & minimum release ratio  when $R_0=0.9$ &minimum release ratio  when $R_0=0.95$ \\
\hline
DFE & $1.007$ & $0.571$&$0.259$\\
\hline
KHA & $0.549$ & $0.317$&$0.151$\\
\hline
KHM &  $0.478$ & $0.244$&$0.107$\\
\hline
KHM2& $0.329$ & $0.165$ &$0.072$\\
\hline
KHMA & $0.210$ & $0.112$ &$0.051$\\
\hline
\end{tabular}
\caption{Different  population suppression strategies applied before release of  \emph{Wolbachia}  infected female mosquitoes  can reduce the minimum number of \emph{Wolbachia}-infected mosquitoes that can lead to persistence of \emph{Wolbachia}. The top row indicates that you have to release $1.007$ times as many infected female mosquitoes as there are wild infect mosquitoes to establish an infection in the wild that starts at the DFE when $R_0=0.85$. The number of infected mosquitoes that need to be released to establish an infection decreases with the strategies: KHA denotes killing half of the aquatic stage mosquitoes, KHM denotes killing half of adult  wild mosquitoes, KHM2 denotes  killing half adult  of wild mosquitoes once, then kill  half adult  wild mosquitoes again after two weeks, and finally KHMA denotes killing half adult wild mosquitoes and half aquatic stage of mosquitoes.  }
\label{table:thresholds}
\end{table}

 \style{ The results section should provide details of all of the experiments that are required to support the conclusions of the paper. There is no specific word limit for this section, but details of experiments that are peripheral to the main thrust of the article and that detract from the focus of the article should not be included. The section may be divided into subsections, each with a concise subheading. Large datasets, including raw data, should be submitted as supporting files; these are published online alongside the accepted article. The results section should be written in the past tense.}

%********************* Discussion **********************

\section{Discussion}  \label{Discussion}
We developed a  model considering two sex of aquatic stage and adult mosquitoes, diversity in the death rates of \emph{Wolbachia}-infected mosquitoes and \emph{Wolbachia}-free  mosquitoes, and egg laying rates of \emph{Wolbachia}-infected female mosquitoes and \emph{Wolbachia}-free female mosquitoes. The general model is not constrained to particular weather condition, specific  \emph{Wolbachia}  strains,  or specific mosquito species, and it can be easily adapted to \emph{Aedes aegypti}  or \emph{Aedes albopictus} at any location with parameters calibrated using realistic  environmental factors, such as  temperature and rainfall etc.

We found conditions for the existence of multiple equilibria and  backward bifurcation.
If  vertical transmission is complete, then  a unique endemic equilibrium exists  and is not LAS when $R_0<1$ and $R_{0u}>1$.  Backward bifurcation occurs when endemic equilibrium changes into disease free equilibrium if we decrease the initial number of infected individuals, or it reaches another LAS  equilibrium if we increase the number of initially infected individuals.   When vertical transmission is incomplete but more than $50\%$,  and $R_0<1$, two endemic equilibria coexist, but only one endemic equilibrium is LAS such that it becomes disease free equilibrium or endemic equilibrium by perturbation. Since \emph{Wolbachia}-infected mosquitoes are less capable of transmitting dengue virus, complete \emph{Wolbachia}-infection is the ideal case for dengue control.

 When $R_0>1$,  \emph{Wolbachia}  can spread with a small number of initially infected mosquitoes, although very slowly. Population replacement may occur.  When $R_0<1$,  \emph{Wolbachia}  can spread if initial number of infected individuals exceeds a threshold, which depends on vertical transmission rate, ratio of egg laying rates  of  \emph{Wolbachia}-infected females to egg   laying rates of \emph{Wolbachia}-free female mosquitoes, and  ratio of  death rates of infected female mosquitoes to death rates of uninfected female mosquitoes.
 A smaller number of  \emph{Wolbachia}-infected female mosquitoes is needed to be released for persistence of \emph{Wolbachia} if a population suppression strategy is implemented before the  release.

The reproduction number is the product of  the vertical transmission rate, ratio of the egg laying rates of \emph{Wolbachia}-infected mosquitoes to egg laying rates of \emph{Wolbachia}-free mosquitoes, and the ratio of death rates of  \emph{Wolbachia}-free mosquitoes to death rates of  \emph{Wolbachia}  infected mosquitoes.
If the total number of eggs produced by \emph{Wolbachia}-infected mosquitoes through vertical transmission is more than the total number of eggs produced by \emph{Wolbachia}-free mosquitoes, then $R_0 >1$, such that complete infection or endemic equilibrium is LAS. %Since the ratio depends on density, temperature etc., completed infection can be obtained by adjusting the density of mosquitoes.

%The initial condition threshold provides guidance on release strategy of  \emph{Wolbachia}  infected mosquitoes.

The number of    \emph{Wolbachia}-infected mosquitoes required for  sustained \emph{Wolbachia}   infection depends on  vertical transmission rate, $v_w$, the ratio of the number of eggs  laid by \emph{Wolbachia}-infected mosquitoes to the number of eggs laid by \emph{Wolbachia}-free mosquitoes,
$\frac{\phi_w}{\phi_{u}}$,  the ratio of the death rates of  \emph{Wolbachia}-free females to death rates of  \emph{Wolbachia}-infected females, $\frac{\mu_{fu}}{\mu_{fw}}$, and the ratio of the death rates of  \emph{Wolbachia}-free males to death rates of  \emph{Wolbachia}-infected males $\frac{\mu_{mu}}{\mu_{mw}}$. These parameters depend on  pariculiar\emph{Wolbachia}  strain. If the life span of a mosquito is shorter, then the mosquitoes will lay fewer eggs. Once we know the specific  \emph{Wolbachia}  strain that a specific mosquito species, such as  \emph{Aedes aegypti}  or  \emph{Aedes albopictus}   is carrying, we can estimate the number of infected individuals needed to be released for sustainable  \emph{Wolbachia} establishment using this model.

\iffalse \notes{The following is copied from the introduction.  We need to restate it in different words.  MH} \fi

We find that reducing both the  aquatic stage and adult uninfected mosquitoes before releasing \emph{Wolbachia}-infected female mosquitoes is the most effective strategy to reduce the number of \emph{Wolbachia}-infected female mosquitoes needed for  \emph{Wolbachia} persistence (Table \ref{table:thresholds}).  The second most effective strategy was to repeatedly kill the wild uninfected mosquitoes (to reduce both the adult and the aquatic stage mosquitoes) before releasing the infected mosquitoes.

The model and analysis can help in understanding how  \emph{Wolbachia} can invade and persist in mosquito populations. In future research, we will couple our model to a dengue fever transmission model
and analyze the impact that a bacteria infected mosquito population has on the spread of dengue.

 \style{ The discussion should spell out the major conclusions of the work along with some explanation or speculation on the significance of these conclusions. How do the conclusions affect the existing assumptions and models in the field? How can future research build on these observations? What are the key experiments that must be done? The discussion should be concise and tightly argued. The results and discussion may be combined into one section, if desired.}

\section*{Acknowledgments}
 This work was supported by the endowment for the Evelyn and John G. Phillips Distinguished Chair in Mathematics at Tulane University,  the National Science Foundation MPS award DMS-1122666,  and the NIH/NIGMS  program for Models of Infectious Disease Agent Study (MIDAS) award U01GM097661.
 The content is solely the responsibility of the authors and does not necessarily represent the official views of the
 %Department of Energy or the
 National Science Foundation or the National Institutes of Health.
The authors thank  Michael Robert for his many helpful comments.

\section*{Appendix}
The system of equations for mosquito population dynamics  is:
\begin{subequations}
\label{E:mosquitopopulationmodel}
\begin{align}
\frac{\dif A}{\dif t} &=\phi F\left(1-\frac{A}{K_a}\right)-\mu_aA -\psi A\\
\frac{\dif F}{\dif t} &=\theta\psi A-\mu_{fu} F\\
\frac{\dif M}{\dif t} &=(1-\theta)\psi A -\mu_{mu} M
\end{align}
\end{subequations}
The parameters for this model are described in Table \ref{T:statevarandparameters}.

\begin{theorem}

The zero equilibrium for mosquito population dynamics is LAS when $R_{0u}<1$, while the steady state is LAS when $R_{0u}>1$, where $R_{0u}=\frac{\phi \theta \psi}{\mu_{fu}(\mu_a+\psi)}$.
\end{theorem}
\begin{proof}
Jacobian matrix of system (\ref{E:mosquitopopulationmodel}) is:

\begin{equation*}
J=\begin{bmatrix}
-\mu_a-\psi-\phi \frac{F}{K_a}&\phi\left(1-\frac{A}{K_a}\right)& 0\\
\theta\psi&-\mu_{fu}& 0\\
(1-\theta)\psi&0&-\mu_{mu}
\end{bmatrix}.
\end{equation*}
The Jacobian matrix for no-infection equilibrium is:
\begin{equation*}
J_0=\begin{bmatrix}
-\mu_a-\psi&\phi& 0\\
\theta\psi&-\mu_{fu}& 0\\
(1-\theta)\psi&0&-\mu_{mu}
\end{bmatrix}.
\end{equation*}
The characteristic polynomial of $J_0$ is:
\begin{equation}
(\lambda+\mu_{mu})[(\lambda+\mu_a+\psi)(\lambda+\mu_{fu})-\phi\theta\psi]=0
\end{equation}

If $R_{0u}<1$, then all eigenvalues are negative, the zero equilibrium  is LAS.

The Jacobian matrix at  steady state is:
\begin{equation*}
J_{ss}=\begin{bmatrix}
-\mu_a-\psi-\phi \frac{F^*}{K_a}&\phi\left(1-\frac{A^*}{K_a}\right)& 0\\
\theta\psi&-\mu_{fu}& 0\\
(1-\theta)\psi&0&-\mu_{mu}
\end{bmatrix}.
\end{equation*}
Characteristic polynomial of $J_{ss}$ is:
\begin{equation}
(\lambda+\mu_{mu})\left[\lambda+\mu_a+\psi+\frac{\phi\theta\psi}{\mu_{fu}}\left(1-\frac{(\mu_a+\psi)\mu_{fu}}{\phi\theta\psi}\right)\right](\lambda+\mu_{fu})-\phi\theta\psi]=0
\end{equation}

If $R_{0u}>1$, then all eigenvalues of $J_{ss}$ are negative, and the non-zero steady state  is LAS.
\end{proof}
\begin{theorem}
The zero equilibrium: $(0, 0, 0)$ is globally asymptotically stable (GAS) when $R_{0u}<1$.  The steady state: $(K_a(1-\frac{1}{R_{0u}}), \frac{K_a\theta\psi}{\mu_{fu}}(1-\frac{1}{R_{0u}}), \frac{K_a(1-\theta)\psi}{\mu_{mu}}(1-\frac{1}{R_{0u}}))$ is GAS when $R_{0u}>1$.
\end{theorem}

\begin{proof} We follow the approach in \cite{Dumont2012}. First we consider the subsystem:
\begin{subequations} \label{E:mosquitopopulation1}
\begin{align}
\frac{\dif A}{\dif t} &=\phi F\left(1-\frac{A}{K_a}\right)-\mu_aA -\psi A\\
\frac{\dif F}{\dif t} &=\theta\psi A-\mu_{fu} F
%\frac{\dif M}{\dif t} &=(1-\theta)\psi A -\mu_{mu} M
\end{align}
\end{subequations}

where $f=( A, F)$ is $C^1$  on an open set $D \subset \mathbb R^2$. % Note that
%for $\forall x \in D$ and $\delta_n \frac{\partial f_{i}( x_{1}, x_{2})}{\partial x_i}>0$, where $\delta_i \in [-1, +1]$, %$i=1, 2$, $n=2$.

Let $\delta_1=\delta_2=1$, $\delta_1 \frac{df_2}{dx_1} =\theta\psi >0$, $\delta_1 \frac{df_1}{dx_2} =\phi \geq 0$,
 and $\delta_2  \frac{df_2}{dx_2} =\mu_{fu} >0$.
According to the definition of tridiagonal feedback \cite{Dumont2012},  system $\ref{E:mosquitopopulation1}$ is a monotone tridiagonal feedback system with Poincar\'{e}-Bendixson property \cite{Mallet-Paret1996}.

Recall Theorem $2$ in \cite{LiandWang2002}.
If the systems of ODEs $dx/dt=f(x), x \subset  D$ satisfies:
\begin{enumerate}
\item The system exists on a compact absorbing set $K \subset  D $.
\item A unique equilibrium point E exists and is LAS.
\item The system has Poincar\'{e}-Bendixson property.
\item  Each periodic orbit of the system is asymptotically stable.
\end{enumerate}
Then E  is globally asymptotically stable in $D$.

To prove that each periodic orbit $\Omega={p(t):0 \leq t \leq w}$ of system $(\ref{E:mosquitopopulation1})$ is asymptotically stable, we follow \cite{Muldowney1990} and Theorem $3$ in \cite{Dumont2012}.
We need to prove that  the linear system
$\frac{dz(t)}{dt} = J_F^{ [2]}
 (p(t))z(t)$
is asymptotically stable, where  $J_F^{ [2]}$ is the second additive compound matrix of the
Jacobian matrix $J_F$ associated with system $(\ref{E:mosquitopopulation1})$.
For system  $(\ref{E:mosquitopopulation1})$,  $J_F^{ [2]}=-\left(\psi +\mu_a+\frac{\phi}{K_a}F+\mu_{fu}\right)$.
We build  the following linear system with one equation and the right hand side is the compound matrix of the
Jacobian matrix $J_F$.
\begin{align*}
\frac{\dif X}{\dif t} &=-\left(\psi +\mu_a+\frac{\phi}{K_a}F+\mu_{fu}\right)X.
%\frac{\dif Y}{\dif t} &=-\phi F(1-\frac{A}{K_a})-\phi(1-\frac{A}{K_a}) Z\\
%\frac{\dif Z}{\dif t} &=(1-\theta)\psi Y -\mu_{mu} Z
\end{align*}
Let Lyapunov function $V(X, A, Y)=|X|$.

The right derivative of $V$ along the solution paths (X) and $(A, F)$ is:
$D_+(V(t)=-(\psi +\mu_a+\frac{\phi}{K_a}F+\mu_{fu})|X|$, which implies that $V(t) \rightarrow 0$, and $X(t) \rightarrow 0$ as
$t \rightarrow \infty $. Therefore, the linear system \ref{E:mosquitopopulation1} is asymptotically stable, and the solution $(A, F)$ is asymptotically orbitally stable  with asymptotic phase.

By the same argument in   \cite{Dumont2012}, the system  \ref{E:mosquitopopulation1} is uniformly persistent  in $D \subset \mathbb R^2$. The zero equilibrium $(0, 0, 0)$ is isolated and the largest compact invariant outside $D$ is $(A^* , F^*)$, which is absorbing and the system \ref{E:mosquitopopulation1} is uniformly persistent \cite{Hofbauer1989}.  The conditions for Theorem 2 in \cite{LiandWang2002} are all satisfied.
Therefore, $(0,0)$ is GAS whenever $R_{0u}<1$, and $ (A^*, F^*)$ exists  and is GAS when $R_{0u}>1$.

Following  \cite{Vidyasagar1980} and Theorem $4$ in \cite{Dumont2012}, $(0, 0, 0)$ is GAS whenever $R_{0u}<1$ and $(A^*, F^*, M^*)$ is GAS whenever $R_{0u}>1$.
\end{proof}
\begin{theorem}
%The model is well posed under the domain
\begin{equation}\label{E:domwellposed}
\mathcal{D} = \left\{ \begin{pmatrix}
                             A \\F \\ M
                            \end{pmatrix}
                    \in \mathbb{R}^3 \left|
                    \begin{array}{c}
                    0\leq A \leq K_a, \\
                    0 \leq F \leq \frac{\psi \theta K_a}{\mu_{fu}}, \\
                    0 \leq M \leq \frac{\psi (1-\theta)K_a}{\mu_{mu}}
                    \end{array}
                    \right. \right\}.
\end{equation}  \noindent
is an invariant region under the flow induced by (\ref{E:mosquitopopulationmodel}).
\end{theorem}
%The proof is following \cite{Moulay2011}.
\begin{proof}
 The proof directly  follows the proofs for Lemma 4.2 and 4.3 in  \cite{Moulay2011}.
\end{proof}

\bibliographystyle{model1-num-names}
\bibliography{Wolbachia}
\end{document}